\documentclass[11pt]{article}

\usepackage{amsmath, amssymb, amsthm, mathtools}
\usepackage[margin=1in]{geometry}
\usepackage{enumitem}
\usepackage{hyperref}
\usepackage{cleveref}
\usepackage{mleftright}
\usepackage{xcolor}
\hypersetup{
	colorlinks=true,
	linkcolor=blue!50!black, 
	citecolor=green!50!black, 
	urlcolor=red!50!black 
}
\usepackage{algorithm}
\usepackage{algorithmicx}
\usepackage{algpseudocode}

\newtheorem{theorem}{Theorem}[section]
\newtheorem{lemma}[theorem]{Lemma}

\theoremstyle{remark}

\theoremstyle{definition}
\newtheorem{definition}{Definition}

\newcommand{\Sph}{\mathbb{S}}
\newcommand{\R}{\mathbb{R}}
\newcommand{\Z}{\mathbb{Z}}

\newcommand{\eps}{\epsilon}

\newcommand{\cC}{\mathcal{C}}

\newcommand{\cS}{\mathcal{S}}
\newcommand{\ip}[2]{\langle #1,#2\rangle}
\newcommand{\norm}[1]{\left\|#1\right\|}

\newcommand{\op}{\textrm{op}}

\DeclareMathOperator{\poly}{poly}
\DeclareMathOperator{\polylog}{polylog}
\DeclareMathOperator{\range}{range}

\DeclareMathOperator{\sgn}{sgn}
\DeclareMathOperator{\diag}{diag}

\title{A Near-Optimal Lower Bound for $\ell_p$-Subspace Embeddings, $1\leq p<2$}
\author{Yi Li \\ Nanyang Technological University \\ \texttt{yili@ntu.edu.sg}}
\date{}
\begin{document}
\maketitle

\begin{abstract}
	For $d \geq 2$, $p \geq 1$ and $\epsilon > 0$, let $N_p(d,\epsilon)$ be the smallest integer $N$ such that for every integer $n$ and every $A\in\mathbb{R}^{n\times d}$, there exists a matrix $\Phi\in\mathbb{R}^{N\times n}$ satisfying
	$(1-\epsilon)\lVert Ax\rVert_p\leq \lVert\Phi A x\rVert_p\leq (1+\epsilon)\lVert Ax\rVert_p$ for all $x\in\mathbb{R}^d$.
	For every constant $p\geq 1$ with $p\not\in 2\mathbb{Z}$ and all sufficiently small $\epsilon>0$ depending only on $p$, when $d\gtrsim_p \log(1/\epsilon)$, the bound
	\[
		N_p(d,\epsilon) \gtrsim_{p} \frac{d}{\epsilon^2 \operatorname{polylog}(d/\epsilon)}
	\]
	is established. This improves the previous lower bound $\Omega(1/(\epsilon^2\operatorname{polylog}(1/\epsilon)))$ due to Li et al. (SICOMP 2021) and is optimal up to logarithmic factors for $1\leq p<2$.  The central technical idea originated from ChatGPT 6 Pro.
\end{abstract}

\section{Introduction}

The $\ell_p$-subspace embedding problem is a classical problem in the local theory of Banach spaces. In its original form, it considers embedding a finite-dimensional subspace of $L^p[0,1]$. Since such a subspace can be embedded with small distortion into $\ell_p^n$ with a sufficiently large $n$, we present the $\ell_p$-subspace embedding problem in the following form.

\begin{definition}[$\ell_p$-subspace embedding]
	For $d\geq 2$, $p\geq 1$ and $\eps>0$, let $N_p(d,\eps)$ be the smallest integer $N$ such that for every integer $n$ and every $A\in\R^{n\times d}$, there exists a matrix $\Phi \in \R^{N\times n}$ satisfying
	\[
		(1-\epsilon)\norm{Ax}_{p}\leq \norm{\Phi A x}_p \leq (1+\eps)\norm{Ax}_p,	\qquad \forall x\in \R^d.
	\]
\end{definition}

The problem has a long research history, dating back to the 1980s; see a comprehensive survey in~\cite{handbook:19}. Different results in the literature obtain similar upper bounds for $N_p(d,\eps)$ of the form
\begin{equation}\label{eqn:N_p(d,eps)}
	N_p(d,\eps) \leq C \eps^{-2}d^{\max\{1, p/2\}}\log^c (d/\eps)
\end{equation}
but with different logarithmic factors. Here, $C,c > 0$ are absolute constants. The proofs can be found in, e.g., the monograph by Ledoux and Talagrand~\cite{LT91}. It is known that when $p\in 2\Z$, $N_p(d,\eps) \lesssim p^{-2}\eps^{-2}\binom{d+p/2-1}{p/2}$ for $\eps\leq 1/p$, without additional logarithmic factors~\cite{schechtman:tight}. Very recently, Reis and Rothvoss~\cite{RR26} obtained $N_1(d,\eps) \lesssim d/\eps^2$, removing the logarithmic factor for $p = 1$.

For fixed $\eps$, the dependence $d^{\max\{1,p/2\}}$ in \eqref{eqn:N_p(d,eps)} has long been known to be tight. Regarding the dependence on $\eps$, it was proved a few years ago that $\eps^{-2}$ is tight up to logarithmic factors when $d\gtrsim \log(1/\eps)$ and $p\not\in 2\Z$~\cite{LWW21}. Here, some constraint on $d$ is necessary, since it is known that one can do strictly better than $\eps^{-2}$ when $d$ is a constant. In particular, for fixed $d$ and $p\not\in 2\Z$, it is now known that $N_p(d,\eps) \asymp_{d,p} (1/\eps)^{2(d-1)/(d+2p)}$. The lower bound was proved in~\cite{LLW23} and the upper bound was proved very recently in~\cite{L26:const_dim}. The assumption $p\not\in 2\Z$ is also necessary, since isometric embeddings are known to exist for $p\in 2\Z$~\cite{handbook:21}, and so there cannot be a lower bound involving $\eps$.

For $p\in[1,2)$ and $d\gtrsim \log(1/\eps)$, Li et al.~\cite{LWW21} also proved that the factor $\eps^{-2}d$ is tight up to logarithmic factors when the embedding matrix $\Phi$ is restricted to be a row-sampling matrix (with rescaling).
However, it was not known in general whether the combined dependence $\eps^{-2}d^{\max\{1,p/2\}}$ is tight. Our main result gives a lower bound of $\widetilde{\Omega}_p(d/\eps^2)$ for every constant $p\geq 1$ with $p\notin 2\mathbb{Z}$. In particular, for $1\leq p<2$, this matches the known upper bound up to logarithmic factors.

The proof follows the framework of Li et al.~\cite{LWW21}. They introduced a data structure problem called the $\ell_p$-subspace sketch problem, proved a bit lower bound for it, and then converted this bit lower bound to a dimension lower bound. We first define the following for-all version of the $\ell_p$-subspace sketch problem.
\begin{definition}[$\ell_p$-subspace sketch, for-all version]
	Given an $n\times d$ matrix $A$ with entries specified by $O(\log(nd))$ bits,
	an accuracy parameter $\eps  > 0$ and a constant $p\geq 1$, design a data structure
	$Q_p$ so that with probability at least 0.9, $Q_p(x) = (1 \pm  \eps)\|Ax\|_p^p$ for all $x \in \R^d$.
\end{definition}

Throughout the paper, let $\alpha_p=\max\{1,p/2\}$. We have the following lower bound for the $\ell_p$-subspace sketch problem.

\begin{theorem}\label{thm:subspace_sketch_lower_bound}
	Let $p \in  [1, \infty) \setminus 
	2\mathbb{Z}$ be a constant. There exists $\eps_0(p)\in(0,1/4]$ such that, for any $0<\eps\leq\eps_0(p)$,
	$d = \Omega_p(\log(1/\eps))$ and
	$n = \widetilde\Omega_p(d^{\alpha_p}/\eps^{2+2\alpha_p})$, we have that $\Omega_p(d^2/(\eps^2\log^{p+6}(2/\eps)))$ bits are necessary to solve the $\ell_p$-subspace sketch problem.
\end{theorem}

Following the same approach as in \cite{LWW21} to convert a bit lower bound to a dimension lower bound, we obtain the following theorem.
\begin{theorem}\label{thm:dimension_lower_bound}
	Let $p\in [1, \infty) \setminus 
	2\mathbb{Z}$ be a constant. There exists $\eps_0(p)\in(0,1/4]$ such that,
	for $0<\eps\leq\eps_0(p)$ and $d=\Omega_p(\log(1/\eps))$, it holds that
	\[
		N_p(d,\epsilon) \gtrsim_{p} \frac{d}{\eps^2 \polylog(d/\eps)}.
	\]
\end{theorem}
When $p\in [1,2)$, this lower bound is tight up to logarithmic factors and the $\ell_p$-subspace embedding problem is almost settled. For $p>2$ with $p\notin 2\mathbb{Z}$, proving a lower bound of $\widetilde{\Omega}(d^{p/2}/\eps^2)$ for $N_p(d,\eps)$ remains open.

\subsection{Review of the previous construction}\label{sec:review}
Before proving Theorem~\ref{thm:subspace_sketch_lower_bound}, we first recall the strategy underlying the earlier $\tilde{\Omega}(d/\eps^2)$ lower bound of \cite{LWW21} for the for-each $\ell_p$-subspace sketch problem. For easier comparison, we shall describe the construction for a for-all version.

The argument begins with a basic hard instance $A\in \R^{2^k\times k}$, where $k$ is an integer multiple of $4$ and $2^k=\widetilde{\Theta}_p(\eps^{-2})$, for which any $\ell_p$-subspace sketch requires $2^k/\poly(k) = \tilde{\Omega}(\eps^{-2})$ bits. To extend the construction to general dimension $d$, one places independent copies of $A$ along the diagonal of a block-diagonal matrix with $d$ columns. We therefore review the lower-bound argument for this basic instance.

Let the index universe $U = \{-1,1\}^k$. Let $M\in\R^{2^k\times 2^k}$ with $M_{i,j} = |\langle i,j\rangle|^p$ for $i,j\in U$.  The matrix $M$ has $r = \binom{k}{k/2}$ eigenvalues $\lambda$ satisfying $|\lambda| \gtrsim_p \sqrt{2^k/k}$. The normalized Hadamard matrix $H$ diagonalizes $M$; after permuting the columns of $H$, assume that the first $r$ diagonal entries of
$H^\top MH$ are all $\lambda$. Define the spectrum-truncated version of $M$ as
\[
	\widetilde{M} = H \diag\{\underbrace{\lambda,\dots,\lambda}_{r\text{ copies}},0,\dots,0\} H^\top.
\]
The argument of \cite[Lemma 3.6]{LWW21} gives a subset $I\subseteq U$ with $|I| = \Omega(r/\log r)$ and makes rows in $I$ orthogonal by small perturbations within $\range(\widetilde{M})$.
The message encoded in the hard instance is the sign vector $(s_i)_{i\in I}\in\{-1,1\}^I$.	
Given this sign vector, let $x \in \range(\widetilde{M})$ be the linear combination of these normalized orthogonal rows with coefficients $s_i$. For an appropriately chosen $\Delta$, there is a set $\mathcal{S}\subseteq \{-1, 1\}^I$ of size $2^{\Omega(|I|)}$ for which $x_j + \Delta \geq 0$ for all $j \in U$ and each encoded sign is recoverable from a sufficiently accurate estimate of the corresponding coordinate of $Mx$.
For every such sign vector, define weights $y_j=(x_j+\Delta)^{1/p}$ and rows $A_j=y_j\cdot j^\top$ for $j\in U$.

For a query vector $i \in I$, $(Ai)_j = y_j \ip{j}{i}$, and hence 
\[
	\norm{Ai}_p^p = \sum_j (x_j + \Delta)|\ip{j}{i}|^p = (Mx)_i + \Delta(M\mathbf{1})_i.
\]
The second term $\Delta(M\mathbf{1})_i$ is known exactly. The perturbation ensures that $(Mx)_i$ has sign $s_i$ with enough margin to tolerate the norm estimation error. Thus, a query at $i$ recovers the bit $s_i$ by
\[
	s_i = \sgn(Q_p(i) - \Delta(M\mathbf{1})_i).
\]
Distinct sign vectors in $\mathcal{S}$ therefore require distinct contents stored in the data structure. Hence the data structure must have at least
\[
	\Omega(|I|) = \Omega\mleft(\frac{r}{\log r}\mright) = \Omega\mleft(\frac{2^k}{\poly(k)}\mright) = \Omega\mleft(\frac{1}{\eps^2\polylog(1/\eps)}\mright)
\]
bits.

\subsection{Our modification}
The starting point of our construction is the analogous spectral construction for a related matrix $M\in\R^{2^k\times 2^k}$, with $k=4h+1$, $2^k = \widetilde{\Theta}_p(1/\eps^2)$ and $M_{ij} = |\langle i,j\rangle|^{p-2}$. Since $k$ is odd, all inner products $\langle i,j\rangle$ are nonzero. For this new matrix, set $r=\binom{k}{(k-1)/2}$ and let $H$ be the normalized Hadamard matrix whose columns are ordered so that the first $r$ are eigenvectors of $M$ with common eigenvalue $\lambda=\lambda_{k,p}$, as established in Lemma~\ref{lem:eigenvalue}.

The key change is to replace scalar signs by $s\times s$ sign matrices, where $s=(d-k)/2$, and attach one sign matrix $S_i$ to each selected Walsh eigenvector, indexed by $i\in[r]$. Thus, the hard instance encodes $\Theta(s^2)$ bits at each of these $r$ indices.

Next we show how to make these matrix-valued signs visible to norm queries. We use the first $r$ columns of $H$, corresponding to the selected Walsh eigenvectors. Let $S_1,\dots,S_r$ be the collection of sign matrices. For each $j\in U$, let
\[
	X_j
	=
	\sum_{i=1}^r H_{ji} \frac{S_i}{\sqrt{s}}.
\]
Note that the family $\{X_j\}_{j\in U}$ plays the role of the vector $x$ in the previous construction described in Section~\ref{sec:review}.

The hard matrix $A$ is built so that it has $k+2s$ columns and, for $x = (y,a,b)\in\R^k\oplus\R^s\oplus\R^s$, 
\[
	\norm{Ax}_p^p
	\approx 
	F(y,a,b)
\]
for
\begin{equation}\label{eqn:F(y,a,b)}
	F(y,a,b)
	:=
	\sum_{j\in U}  \left(\langle j,y\rangle^2 + \norm{a+\eta X_j b}_2^2 + \norm{b}_2^2\right)^{p/2},
\end{equation}
where $\eta > 0$ is a small constant.

For $i\in U$, $u,v\in\Sph^{s-1}$ and $t\in [0,1/4]$, consider two query vectors
\[
	x_i^+(t) = \left(i,\frac{\sqrt{t}u}{\sqrt 2},\frac{\sqrt{t}v}{\sqrt 2}\right),\quad 
	x_i^-(t) = \left(i,-\frac{\sqrt{t}u}{\sqrt 2},\frac{\sqrt{t}v}{\sqrt 2}\right).
\]
Then
\[
	F(x_i^\pm(t)) = \sum_{j\in U} \mleft(\langle i,j\rangle^2 + t G_j^\pm\mright)^{p/2},\quad\text{where}\quad G_j^\pm = 1 \pm\eta u^\top X_j v + \frac{\eta^2}{2}\norm{X_j v}_2^2.
\]
Let $F_i^\pm(t) = F(x_i^\pm(t))$. Then
\[
	(F_i^+)'(0) - (F_i^-)'(0) = p\eta u^\top Y_i v.
\]
To estimate a derivative, we can use Taylor expansion up to a sufficient degree. This means, from the query vectors $x_i^+$ and $x_i^-$, we can approximately recover the bilinear form
\[
	u^\top Y_i v,\quad \text{where}\quad Y_i = \sum_j |\langle i,j\rangle|^{p-2} X_j = \sum_j M_{ij} X_j,
\]
where $Y_i$ is the analogue of $(Mx)_i$ in the old construction described in Section~\ref{sec:review}. 
Here we have the exact identity $Y_i = \lambda X_i$. 
This further implies that we can recover, from $x_i^+$ and $x_i^-$, the bilinear form $u^\top X_i v$ up to a small constant additive error. To reconstruct the sign matrices $S_i$, we choose matrices $\widehat{X}_i$ that are consistent with these values and round the entries of
\[
	\sqrt{s}\sum_{j\in U} H_{ji}\widehat{X}_j
\]
to signs to obtain $\widehat{S}_i$ for $i\in [r]$. One can show that the reconstructed collection agrees with the true collection except on a constant fraction of the $rs^2$ entries.

We choose the sign matrix collections $(S_1,\dots,S_{r})$ to have large pairwise total Hamming distance, which allows us to decode the true collection from the reconstructed collection $(\hat{S}_1,\dots,\hat{S}_{r})$ as the unique nearest neighbour. There exists a family of $2^{\Omega(r s^2)}$ such collections. Since $r = \widetilde{\Omega}_p(\eps^{-2})$ and $s=\Theta(d)$, this gives the desired $\Omega(rs^2) = \widetilde{\Omega}_p(d^2/\eps^2)$ bit lower bound.

\section{Hard Instance}
As previewed, we shall replace signs $s_i$ in the construction of \cite{LWW21} with sign matrices $S_i$ and aim to recover the sign matrices $S_i$. Fix an absolute constant $\rho\in(0,1/4)$ and set $\theta=\sqrt{\rho}/8$.
Let $C_0\geq1$ be the absolute constant from Lemma~\ref{lem:S_size}, and set $\eta=1/(10C_0)$.
Let $\Lambda = \log(1/\eps) + 1$ and choose $k=4h+1$ to be the largest integer of this form satisfying, for a sufficiently small constant $c_{p,\theta}>0$ depending only on $p$ and $\theta$,
\begin{equation}\label{eqn:assumption_k}
	2^k \leq \frac{c_{p,\theta}}{\eps^2 \Lambda^4 k^{p+3/2}}.
\end{equation}
Choose $\eps_0(p)>0$ sufficiently small that, whenever $0<\eps\leq\eps_0(p)$, the selected $k$ satisfies $k\geq5$ and $k>2p-3$.
Let $s = \lfloor (d-k)/2\rfloor$. Our query space is $\R^{k+2s} = \R^k \oplus \R^{2s}$. If $k + 2s < d$, we can append a zero column. To simplify the notation, we assume that $d = k+2s$ below. By choosing the implicit constant in $d\gtrsim_p \log(1/\eps)$ sufficiently large, we ensure that $s\geq k$ and $s = \Theta(d)$.

Let $M\in\R^{2^k\times 2^k}$ be defined as $M_{i,j} = |\ip{i}{j}|^{p-2}$ for $i,j\in U$. Since $k$ is odd, $\ip{i}{j}\neq 0$ for all $i,j\in U$, so $M$ is well defined. Let $m = (k-1)/2$. 	Similarly to \cite{LWW21}, we have the following property of $M$.

\begin{lemma}\label{lem:eigenvalue}
Let $p\in [1,\infty)\setminus 2\mathbb{Z}$ and $k=4h+1\geq5$ with $k>2p-3$. The matrix $M$ has
\[
	r = \binom{k}{m}
\]
eigenvalues equal to $\lambda_{k,p}$, where
\[
	|\lambda_{k,p}| \gtrsim_p \sqrt{\frac{2^k}{k}}.
\]
\end{lemma}

The normalized Hadamard matrix $H\in \R^{2^k\times 2^k}$ diagonalizes $M$; after permuting
the columns of $H$, assume that the first $r$ diagonal entries of $H^\top M H$ are all
$\lambda_{k,p}$. Define the spectrum-truncated version of $M$ as
\[
	\widetilde{M} = H\diag\{\underbrace{\lambda_{k,p},\dots,\lambda_{k,p}}_{r\text{ copies}}, 0,\dots,0\}H^\top.
\]
It is clear that
\begin{equation}\label{eqn:tilde{M}_i norm}
	\beta := \norm{\tilde{M}_i}_2 = |\lambda_{k,p}|\sqrt{\frac{r}{2^k}},\quad i \in U.
\end{equation}

For arbitrary matrices $W_1,\dots,W_r\in \R^{s\times s}$, define the inverse Walsh-Hadamard transform
\[
	X_j = \sum_{i=1}^r H_{ji} W_i,\qquad j\in U.
\]
This is an analogy of the Walsh transformation for scalars $w_i$.

For each collection of sign matrices $S_1,\dots,S_r\in\{-1,1\}^{s\times s}$, set $W_i=S_i/\sqrt{s}$ and define $X_j$ by the transform above. Define the good family of sign matrix collections
\[
	\cS = \left\{ (S_1,\dots,S_r) : \max_{j\in U} \norm{X_j}_{\op} \leq C_0 \right\}.
\]
\begin{lemma}\label{lem:S_size}
	There exists an absolute constant $C_0 \geq 1$ such that $|\cS|\geq 0.9\cdot 2^{rs^2}$.
\end{lemma}

The following lemma is a corollary of \cite[Theorem 2.3]{PVZ17} and \cite[Proposition 2.4]{PVZ17}.
\begin{lemma}[{\cite{PVZ17}}]\label{lem:T}
	Let $s\geq 1$ and $\eps\in (0,1/4]$. There exists a matrix $T\in \R^{m_0 \times (2s+1)}$, where $m_0 \lesssim_p (s/\eps^{2})^{\max\{1,p/2\}}$ such that
	\[
	(1-\eps)\norm{y}_2^p \leq \norm{Ty}_p^p \leq (1+\eps)\norm{y}_2^p
	\]
	for all $y\in\R^{2s+1}$. Moreover, $\norm{T_{\ast,j}}_p \leq 2$ for all $j\in [2s+1]$.
\end{lemma}

Let $T$ be a matrix guaranteed by Lemma~\ref{lem:T}. Write $T_\ell = (g_\ell, h_\ell^\top, q_\ell^\top)\in \R\oplus \R^{s}\oplus \R^s$ for $\ell \in [m_0]$. For each $S\in\mathcal{S}$, form a matrix $A=A_S$ of $2^k m_0$ rows, indexed by $(i,\ell)\in U\times [m_0]$, by	
\[
	A_{i,\ell} = (g_\ell i^\top, h_\ell^\top, q_\ell^\top + \eta h_\ell^\top X_i) \in \R^{k}\oplus \R^{s}\oplus \R^s.
\]
The hard matrix $\widetilde{A}$ is obtained from $A$ by the rounding described in Section~\ref{sec:rounding}.

\subsection{Analysis of Matrix \texorpdfstring{$M$}{M}}
We prove Lemma~\ref{lem:eigenvalue} in this section.

\begin{proof}[Proof of Lemma~\ref{lem:eigenvalue}]
	Recall that $k=4h+1$ and $m=2h$. 
	The same argument as in \cite{LWW21} shows that there are $\binom{k}{m}$ eigenvalues of the same value and so we only prove the lower bound here. 
	
	Let $q = p-2$. We first assume that $p > 1$, so $q > -1$. Following the coefficient calculation in \cite[Lemma 3.2]{LWW21} with the odd-degree polynomial identity 
	\[
		(1 - x)^{2h} (1 + x)^{2h + 1} = (1 - x^2)^{2h} (1 + x)
	\]
	in place of $(1-x)^{2h} (1+x)^{2h} = (1-x^2)^{2h}$, we obtain that
	\begin{equation}\label{eqn:lambda_expression}
		\lambda_{k,p} = \sum_{i=0}^{2h} (-1)^i \binom{2h}{i} (|4h+1-4i|^q + |4h-1-4i|^q).
	\end{equation}
	Following the same argument in \cite{LWW21}, we have for $-1 < q < 2h$, 
	\[
		\lambda_{k,p} = (-1)^h c_q 2^k I_q,
	\]
	where
	\[
		c_q = -\frac{2\Gamma(q+1)}{\pi}\sin\frac{\pi q}{2},\qquad I_q = \int_0^\infty \frac{\sin^{2h}t \cos^{2h+1}t}{t^{q+1}}\, dt.
	\]
	It then suffices to lower bound the integral $I_q$. Note that
	\[
		I_q = \sum_{j=0}^\infty (-1)^j I_q^{(j)},
	\]
	where
	\[
		I_q^{(j)} = \int_0^{\pi/2} \left(\frac{1}{(j\pi+t)^{q+1}} - \frac{1}{((j+1)\pi - t)^{q+1}}\right) \sin^{2h}t \cos^{2h+1}t \, dt.
	\]
	Since $I_q^{(j)} > 0$ and $I_q^{(j)}$ is decreasing in $j$, we know that $I_q\geq I_q^{(0)} - I_q^{(1)}$. Write
	\[
		I_q^{(0)} - I_q^{(1)} = \int_0^{\pi/2} g(t) \sin^{2h}t \cos^{2h+1}t\, dt,
	\]
	where
	\[
		g(t) = \frac{1}{t^{q+1}} - \frac{1}{(\pi-t)^{q+1}} - \frac{1}{(\pi+t)^{q+1}} + \frac{1}{(2\pi-t)^{q+1}},\quad t\in \left[0,\frac{\pi}{2}\right].
	\]
	On the interval $[\pi/4-c_1/\sqrt{h}, \pi/4+c_1/\sqrt{h}]$, we have
	\[
		\sin^{2h}t \cos^{2h+1}t \geq \frac{c_2}{2^{2h}},\quad g(t)\geq c_{3,q}.
	\]
	Therefore,
	\[
		I_q \gtrsim_q \frac{2^{-k/2}}{\sqrt k}
	\]
	and so for $p>1$,
	\begin{equation}\label{eqn:lambda_lb}
		|\lambda_{k,p}| \gtrsim_p \frac{2^{k/2}}{\sqrt{k}}.
	\end{equation}
	
	Now we deal with $p=1$ (i.e.~$q=-1$), which follows from a limiting argument. Taking $p\searrow 1$ (i.e.~$q\searrow -1$) in \eqref{eqn:lambda_expression}, since it is a finite sum and is continuous in $q$, we have $\lambda_{k,p}\to \lambda_{k,1}$. One can show that $g(t)\geq c_4(q+1)$ for some absolute constant $c_4$ on the small neighbourhood around $\pi/4$, and it is easy to verify using the expression of $c_q$ that $c_q\asymp 1/(q+1)$. Hence, the hidden constant in \eqref{eqn:lambda_lb} can be made an absolute constant when $p$ is close to $1$. This proves the lemma for $p=1$.
\end{proof}

\subsection{Good Sign Matrices}

We prove Lemma~\ref{lem:S_size} in this section.

\begin{proof}[Proof of Lemma~\ref{lem:S_size}]
Choose all sign entries independently and uniformly. For any fixed coefficient vector $b$, the matrix
\[
	Z = \frac{1}{\sqrt s} \sum_i b_i S_i
\]
has independent, mean-zero subgaussian entries, each with subgaussian norm at most $C\|b\|_2/\sqrt s$. By the standard subgaussian matrix norm bound, there exists a sufficiently large absolute constant $C_0\geq 1$ such that
\begin{equation}\label{eqn:operator_tail}
	\Pr\left\{\norm{Z}_{op} > C_0\norm{b}_2\right\} \leq 2e^{-8s}.
\end{equation}
	
Since $H_{ji}^2=2^{-k}$ for every entry $H_{ji}$, we have $\sum_{i=1}^r H_{ji}^2 = r/2^k < 1$.
Using the tail inequality \eqref{eqn:operator_tail}, we obtain that
\[
	\Pr\{ \norm{X_j}_{op} > C_0 \} \leq 2e^{-8s}.
\]
Taking a union bound over $j\in U$, and using $s\geq k$, we see that
\[
	\max_{j\in U} \norm{X_j}_{\op} \leq C_0
\]
except with probability at most
\[
	2\cdot 2^k e^{-8s} < \frac{1}{10}.
\]
This finishes the proof.
\end{proof}

\subsection{Rounding}\label{sec:rounding}

We use the same entrywise perturbation principle as in the proof of \cite[Theorem 5.1]{LWW21}. We first note that, for $x = (v, a, b) \in \R^k \oplus \R^{2s}$,
\begin{equation}\label{eqn:A_lower_bound}
	\norm{Ax}_p^p \geq c_p 2^k \norm{x}_2^p.
\end{equation}
By the definition of $A$ and the embedding guarantee for $T$,
\[
	\norm{Ax}_p^p \asymp \sum_{j\in U} \left(\langle j,v\rangle^2 + \norm{a+\eta X_j b}_2^2 + \norm{b}_2^2\right)^{p/2}
\]
Recall that on the good family $\mathcal{S}$, we have $\norm{X_j}_{op} \leq C_0$. It follows that
\[
	\norm{a+\eta X_j b}_2^2 + \norm{b}_2^2 \asymp \norm{a}_2^2 + \norm{b}_2^2
\]
and so
\begin{equation}\label{eqn:norm_aux_equivalence}
	\norm{Ax}_p^p \asymp \sum_{j\in U} (\langle j,v\rangle^2 + \norm{a}_2^2 + \norm{b}_2^2)^{p/2}.
\end{equation}
By the lower bound of Khintchine's inequality, we have $\sum_{j\in U}|\ip{j}{v}|^p\geq c_p 2^k \norm{v}_2^p$. Therefore,
\[
	\norm{Ax}_p^p \gtrsim 2^k(\norm{v}_2^2+\norm{a}_2^2+\norm{b}_2^2)^{p/2},
\]

Fix a sign matrix collection $S\in \cS$. Round every entry of $A$ to a common grid of dyadic spacing $\zeta$, resulting in $\widetilde{A}$. Since $A$ has $n = 2^k m_0$ rows,
\[
	\norm{(A-\widetilde{A})x}_p \leq n^{1/p}\zeta\sqrt{d}\norm{x}_2.
\]
Choosing
\[
	\zeta \asymp_p \frac{\eps}{m_0^{1/p}\sqrt{d}},
	\qquad
	\zeta \leq c_p'\frac{\eps}{m_0^{1/p}\sqrt{d}},
\]
with $c_p'>0$ sufficiently small gives
\[
	\norm{(A-\widetilde{A})x}_p \leq C_p'\eps \norm{Ax}_p.
\]
Thus, after adjusting the constant in the choice of $\zeta$,
\[
	\norm{\widetilde{A}x}_p^p = (1\pm \eps)\norm{Ax}_p^p, \quad \forall x\in\R^d.
\]
Now we examine the magnitude of the entries of $A$. 
The column bounds in Lemma~\ref{lem:T} give $|g_\ell|,|q_{\ell,t}|\leq 2$ for $\ell\in[m_0]$ and $t\in[s]$. Its embedding guarantee also gives
\[
	\norm{h_\ell}_2 = \sup_{\norm{a}_2 = 1}|(T(0,a,0))_\ell| \leq (1+\eps)^{1/p} \leq 2.
\]
Thus $|q_{\ell,t} + \eta(h_\ell^\top X_j)_t|\leq 2 + 2\eta C_0$ for $j\in U$, $\ell\in[m_0]$ and $t\in[s]$, and so every entry of $A$ is bounded by an absolute constant.
Since $\zeta^{-1}=\poly(d/\eps)$, the rounded entries use $O_p(\log(nd))$ bits.

\section{Recovery of Sign Matrices}

For $S=(S_1,\dots,S_r)\in\cS$, $(y,a,b)\in \R^k\oplus\R^s\oplus\R^s$, define $F_S(y,a,b)$ to be the function defined in \eqref{eqn:F(y,a,b)}, that is,
\[
	F_S(y,a,b)
	= \sum_{j\in U} \left(\langle j,y\rangle^2 + \norm{a+\eta X_j b}_2^2 + \norm{b}_2^2\right)^{p/2}
\]
All quantities $W_i,X_j,Y_i$ and $A$ are understood to be constructed
from the displayed collection $S$; when needed, we write $Y_i^S$ or $A_S$ to
emphasize this dependence.
Fix $S\in\cS$ in the following display and in Lemma~\ref{lem:F_properties}, and let $A=A_S$ be the corresponding unrounded matrix.
When this $S$ is clear from the context, we write $F$ instead of $F_S$. Then
\begin{equation}\label{eqn:A_approximates_F}
	(1-\eps)F(y,a,b)
	\leq 
	\norm{ A (y,a,b) }_p^p 
	\leq (1+\eps)F(y,a,b).
\end{equation}
Hence, we shall first analyse some properties of the function $F$.

For $i\in U$, $u,v\in\Sph^{s-1}$, and $0\leq t\leq1/4$, write
\begin{equation}\label{eqn:paired_queries}
	x_i^\pm(t;u,v) = \left(i,\frac{\pm\sqrt t\,u}{\sqrt2}, \frac{\sqrt t\,v}{\sqrt2}\right),
	\qquad F_i^\pm(t)=F_S(x_i^\pm(t;u,v)).
\end{equation}
The functions $F_i^\pm$ depend on the fixed $S,u,v$, which we suppress
from the notation.

\begin{lemma}\label{lem:F_properties}
For every $S\in\cS$, $i\in U$ and $u,v\in\Sph^{s-1}$,
\begin{equation}\label{eqn:F'_difference}
	(F_i^+)'(0) - (F_i^-)'(0) = p\eta u^\top Y_i^S v.
\end{equation}
Moreover, for every $0\leq t\leq1/4$,
\begin{equation}\label{eqn:query_size}
	F_i^\pm(t)\leq C_p 2^k k^{p/2}.
\end{equation}
\end{lemma}
\begin{proof}
	By the definition of $F_S$,
	\begin{equation}\label{eqn:affine_power_form}
		F_i^\pm(t) = \sum_{j\in U} (\ip{i}{j}^2 + t G_j^\pm)^{p/2},
		\qquad
		G_j^\pm=1\pm\eta u^\top X_j v + \frac{\eta^2}{2}\norm{X_j v}_2^2.
	\end{equation}
	Since $\eta\norm{X_j}_{\op}\leq1/10$, these coefficients lie in
	$[3/4,5/4]$. Also $|\ip{i}{j}|\geq 1$ because $k$ is odd, so all summands
	are differentiable in a neighbourhood of zero. Hence
	\begin{align*}
		(F_i^+)'(0) - (F_i^-)'(0)
		&= \frac p2\sum_{j\in U}|\ip{i}{j}|^{p-2}(G_j^+ - G_j^-) \\
		&= p\eta\sum_{j\in U} M_{ij} u^\top X_jv
		= p\eta u^\top Y_iv.
	\end{align*}

	Now we control $|F_i^\pm(t)|$. Using \eqref{eqn:norm_aux_equivalence} and the upper bound of Khintchine's inequality, we have 
	\[
		F_S(y,a,b) \asymp_p \norm{A(y,a,b)}_p^p \lesssim_p 2^k(\norm{y}_2^2+\norm{a}_2^2+\norm{b}_2^2)^{p/2} = 2^k\norm{(y,a,b)}_2^p
	\]
	and so
	\[
		|F_i^\pm(t)| \lesssim_p 2^k\norm{x_i^\pm(t)}_2^p = 2^k (k+t)^{p/2} \lesssim_p 2^k k^{p/2}. \qedhere
	\]
\end{proof}

Next we shall show how to recover sign matrices $\{S_i\}_{i\in [r]}$ from the bilinear forms $u^\top W_i v$, provided
accurate estimates of $u^\top Y_i v$, and then show how we can obtain such accurate estimates.

Let $S = \{S_i\}_{i\in [r]}$ and $S' = \{S_i'\}_{i\in [r]}$ be two collections of sign matrices in $\cS$. We define their Hamming distance as
\[
	d_H(S,S') = \sum_{i=1}^r d_H(S_i, S_i').
\]
The first stage approximates every matrix $X_i$ and then obtains $\widehat{W}_i\approx W_i = S_i/\sqrt{s}$. Rounding the entries of $\sqrt{s}\widehat{W}_i$ to $\{-1,1\}$ gives a sign matrix $\widehat{S}_i$. One can show that $(\widehat{S}_1,\dots,\widehat{S}_r)$ has a small total Hamming error, which can be corrected to recover the true $(S_1,\dots,S_r)$ if the collection of $(S_1,\dots,S_r)$ we shall use is separated in total Hamming distance.

\begin{lemma}\label{lem:conditional_recovery}
	Let $\cC\subseteq\cS$ be a family of sign matrix collections with
	\begin{equation}\label{eqn:global_separation}
		d_H(S,S')\geq \rho r s^2
		\qquad(S,S'\in\cC,\ S\neq S').
	\end{equation}
	For each $i\in U$, let
	$\widehat D_i:\Sph^{s-1}\times\Sph^{s-1}\to\R$ be a function.
	We say $S\in\cC$ is compatible with these functions if
	\begin{equation}\label{eqn:D_hat}
		\left|\widehat D_i(u,v)-u^\top Y_i^S v\right|
		\leq\theta\beta,
		\qquad \forall\, i\in U,\ \forall\,u,v\in\Sph^{s-1}.
	\end{equation}
	For a compatible collection $S$, these functions determine sign matrices
	$\widehat S_i$ for $i\in[r]$ such that, writing $\widehat S=(\widehat S_1,\dots,\widehat S_r)$,
	\begin{equation}\label{eqn:block_hamming_recovery}
		d_H(\widehat S,S)\leq 16\theta^2 rs^2 = \frac{1}{4}\rho rs^2.
	\end{equation}
	The collection $S$ is the unique nearest codeword in $\cC$ to $\widehat S$ in the total Hamming distance.
	In particular, at most one $S\in\cC$ is compatible with these functions.
\end{lemma}
\begin{proof}
	Let $\mathcal{N}$ be a $(1/4)$-net of $\Sph^{s-1}$. For $i\in U$, one can use linear programming to find a matrix $\widehat{X}_i$ satisfying
	\begin{equation}\label{eqn:normalized_measurements}
		\left|\frac{\widehat D_i(u,v)}{\lambda_{k,p}} - u^\top \widehat{X}_i v\right|\leq \frac{\theta\beta}{|\lambda_{k,p}|} = \theta\sqrt{\frac{r}{2^k}}, \qquad \forall\, u,v\in \mathcal{N}.
	\end{equation}
	The first $r$ columns of $H$ are eigenvectors of $M$ with eigenvalue $\lambda_{k,p}$, so substituting the definition of $X_j$ gives $Y_i=\lambda_{k,p}X_i$. Thus the true $X_i$ satisfies \eqref{eqn:normalized_measurements}. By a standard discretization bound (see, e.g., \cite[Exercise 4.4.3]{Vershynin2018})
	\begin{equation}\label{eqn:compatible_sign_matrix}
		\norm{X_i - \widehat{X}_i}_{\op} \leq 2\max_{u,v\in \mathcal{N}} |u^\top(X_i - \widehat{X}_i)v| \leq 4\theta\sqrt{\frac{r}{2^k}}.
	\end{equation}
	Since $H^\top H=I$, we have $W_i=\sum_{j\in U}H_{ji}X_j$ for $i\in[r]$. Let $\widehat{W}_i = \sum_{j\in U} H_{ji}\widehat{X}_j$ and round the entries of $\sqrt{s}\widehat{W}_i$ to $\{-1,1\}$, obtaining $\widehat{S}_i$. Each incorrectly rounded entry contributes at least $1/s$ to $\norm{W_i-\widehat W_i}_F^2$. Applying Parseval's identity entrywise to $X_j-\widehat X_j$, we obtain
	\begin{align*}
		d_H(\widehat{S}, S)
		&\leq s\sum_{i=1}^r \norm{W_i - \widehat{W}_i}_F^2 \\
		&\leq s\sum_{j\in U} \norm{X_j - \widehat{X}_j}_F^2 \\
		&\leq s^2\sum_{j\in U} \norm{X_j - \widehat{X}_j}_{\op}^2 \\
		&\leq s^2\cdot 2^k\cdot 16\theta^2\frac{r}{2^k}
		=16\theta^2rs^2=\frac{1}{4}\rho rs^2.
	\end{align*}
	Therefore, the nearest codeword is unique and recovers $S$ exactly. 
\end{proof}

\begin{lemma}\label{lem:derivative_approximation}
	Suppose that an instance of $Q_p$ is a correct $\ell_p$-subspace sketch for the rounded hard matrix $\widetilde{A}$.
	Then one can construct functions $\widehat D_i(u,v)$ satisfying \eqref{eqn:D_hat} for all $i\in U$ and all $u,v\in\Sph^{s-1}$.
\end{lemma}
\begin{proof}
	By correctness of $Q_p$, \eqref{eqn:A_approximates_F} and the rounding error, there exists a constant $C_p$ depending only on $p$ such that
	\[
		|Q_p(x)-F_S(x)|\leq C_p\eps F_S(x), \qquad \forall x\in\R^d.
	\]
	Fix $i,u,v$. Let $\widehat{f}_i^\pm(t) = Q_p(x_i^\pm(t; u,v))$, then it follows from Lemma~\ref{lem:F_properties} that
	\[
		|\widehat{f}_i^\pm(t) - F_i^\pm(t)| \leq C_p \eps 2^k k^{p/2},\quad 0\leq t\leq \frac{1}{4}.
	\]
	Set $A_0=C_p2^k k^{p/2}$, increasing $C_p$ if necessary so that
	\[
		\sum_{j\in U}|\ip{i}{j}|^p\leq A_0
		\quad\text{and}\quad
		|\widehat f_i^\pm(t)-F_i^\pm(t)|\leq\eps A_0,
		\qquad 0\leq t\leq1/4.
	\]
	The first inequality follows from the upper bound of Khintchine's inequality, and the second follows from the preceding error bound.
	
	Recall that $F_i^\pm$ have the form
	\[
		F_i^\pm(t) = \sum_{j\in U} (a_j + c_j^\pm t)^{p/2},
	\]
	where $a_j = \ip{i}{j}^2 \geq 1$ and $c_j^\pm \in [\frac{3}{4},\frac{5}{4}]$. Thus, applying Lemma~\ref{lem:derivative_approx} below with this $A_0$, we obtain polynomials $P_i^\pm$ of degree $O_p(\log(1/\eps))$ such that 
	\[
		|(P_i^\pm)'(0) - (F_i^\pm)'(0)| \leq C_p\eps\log^2(2/\eps) 2^k k^{p/2}.
	\]
	Define the estimator
	\begin{equation}\label{eqn:query_estimator}
		\widehat{D}_i(u,v) = \frac{(P_i^+)'(0) - (P_i^-)'(0)}{p\eta}.
	\end{equation}
	Then by \eqref{eqn:F'_difference} and the triangle inequality
	\begin{align*}
		|\widehat{D}_i(u,v) - u^\top Y_i^S v| &= \left|\frac{(P_i^+)'(0) - (P_i^-)'(0)}{p\eta} - \frac{(F_i^+)'(0) - (F_i^-)'(0)}{p\eta}\right| \\
		&\leq C_p\eps\log^2(2/\eps) 2^k k^{p/2}.
	\end{align*}
	By Lemma~\ref{lem:eigenvalue} and $r \asymp 2^k/\sqrt k$,
	\[
		\beta = |\lambda_{k,p}|\sqrt{\frac r{2^k}} \gtrsim_p \frac{2^{k/2}}{k^{3/4}}.
	\]
	Consequently, \eqref{eqn:assumption_k} implies
	\[
		\frac{|\widehat{D}_i(u,v) - u^\top Y_i^S v|}{\beta} \lesssim_p \frac{\eps\log^2(2/\eps)\,2^k k^{p/2}}
		{\beta}
		\lesssim_p \eps\log^2(2/\eps)\,2^{k/2} k^{p/2+3/4}
		\lesssim_p \sqrt{c_{p,\theta}},
	\]
	which can be made at most $\theta$ by taking $c_{p,\theta}$ to be sufficiently small. This proves \eqref{eqn:D_hat}. 
\end{proof}

\subsection{Approximation of Derivative}\label{sec:derivative_approx}
We need a classical inequality on polynomials due to Markov, see, e.g. \cite[p.~218, Equation (32)]{timan}. The form we use follows by rescaling the domain and centering the range.
\begin{lemma}\label{lem:markov}
	Let $Q$ be a degree-$n$ polynomial that maps $[a,b]$ to $[c,d]$. Then 
	\[
		\max_{t\in [a,b]} |Q'(t)| \leq \frac{d-c}{b-a}n^2.
	\]
\end{lemma}

Now we present the main result of this section, which shows that one can estimate the derivative $f'$ using approximate values of $f$.
\begin{lemma}\label{lem:derivative_approx}
	Let $0 < \eps \leq 1/4$,
	\[
		f(t) = \sum_{j=1}^{N} (a_j + c_j t)^{p/2}, \qquad a_j \geq 1,\quad \frac{3}{4}\leq c_j \leq \frac{5}{4},
	\]
	and $A_0 \geq \sum_j a_j^{p/2}$.
	There is an integer $L = O_p(\log(2/\eps))$ and a fixed equally spaced grid $\mathcal{T}_L\subset[0,1/4]$ of $4L^2 + 1$ points with the following property. Given values $\widehat{f}(t)$ satisfying
	\[
		|\widehat{f}(t)-f(t)| \leq \eps A_0, \qquad t\in\mathcal{T}_L,
	\]
	one can choose a polynomial $P$ of degree at most $L$, using only $\eps$, $A_0$ and the values $\{\widehat f(t)\}_{t\in \mathcal{T}_L}$, such that
	\begin{equation}\label{eqn:derivative_error}
		|P'(0)-f'(0)| \leq C_p \eps A_0  \log^2\frac{2}{\eps} .
	\end{equation}
\end{lemma}
\begin{proof}
	Let $P_L$ be the degree-$L$ Taylor polynomial of $f$ at $0$. Since $0\leq c_jt/a_j \leq 5/16$ on $[0,1/4]$, the binomial series gives
	\[
		\norm{f - P_L}_{L^\infty[0,\frac{1}{4}]} \leq C_p 2^{-L}\sum_j a_j^{p/2}, \qquad P_L'(0) = f'(0).
	\]
	Take $L$ large enough such that $C_p 2^{-L}\leq \eps$ and let
	\[
		\mathcal{T}_L = \left\{\frac{\ell}{16L^2}: \ell = 0,1,\dots,4L^2\right\}.
	\]
	Then
	\[
		\max_{t\in\mathcal T_L}|P_L(t)-\widehat{f}(t)| \leq \max_{t\in\mathcal{T}_L}|P_L(t) - f(t)| + \max_{t\in\mathcal T_L}|f(t)-\widehat{f}(t)| \leq \eps A_0 + \eps A_0 = 2\eps A_0.
	\]
	Choose any polynomial $P$ of degree at most $L$ that satisfies
	\begin{equation}\label{eqn:polynomial_fit}
		\max_{t\in\mathcal{T}_L}|P(t)-\widehat f(t)|\leq2\eps A_0.
	\end{equation}
	Such a polynomial exists: $P_L$ satisfies \eqref{eqn:polynomial_fit}. 
	One can also construct a $P$ using the values $\{\widehat{f}(t)\}_{t\in\mathcal{T}_L}$ by solving a linear programme.
	The following argument applies to every feasible choice of $P$.
	
	For any polynomial $Q$ of degree at most $L$, choose $t^\ast\in[0,1/4]$ at which $|Q|$ attains its maximum. Choose a grid point $\hat{t}\in\mathcal{T}_L$ with $|t^\ast - \hat{t}| \leq 1/(16L^2)$. Then by Lemma~\ref{lem:markov},
	\[
		\norm{Q'}_{L^\infty[0,\frac{1}{4}]} \leq 8L^2\norm{Q}_{L^\infty[0,\frac{1}{4}]}.
	\]
	It follows that
	\begin{align*}
		\norm{Q}_{L^\infty[0,\frac{1}{4}]} &= |Q(t^\ast)| \\
		&\leq |Q(\hat{t})| + |Q(t^\ast) - Q(\hat t)| \\
		&\leq \max_{t\in\mathcal{T}_L} |Q(t)| + \frac{1}{16L^2}\cdot 8L^2\cdot \norm{Q}_{L^\infty[0,\frac{1}{4}]}\\
		&= \max_{t\in\mathcal{T}_L} |Q(t)| + \frac{1}{2}\norm{Q}_{L^\infty[0,\frac{1}{4}]}
	\end{align*}
	and
	\[
		\norm{Q}_{L^\infty[0,\frac{1}{4}]}\leq 2\max_{t\in\mathcal{T}_L} |Q(t)|.
	\]
	Hence,
	\begin{align*}
		\norm{P - P_L}_{L^\infty[0,\frac{1}{4}]} &\leq 2\max_{t\in\mathcal{T}_L} |P(t) - P_L(t)| \\
		&\leq 2\left(\max_{t\in\mathcal{T}_L} |P(t) - \hat{f}(t)| + \max_{t\in\mathcal{T}_L} |\hat{f}(t) - f(t)| + \max_{t\in\mathcal{T}_L} |f(t) - P_L(t)|\right) \\
		&\leq 2(2\eps A_0 + \eps A_0 + \eps A_0)\\
		&= 8 \eps A_0.
	\end{align*}
	Applying Lemma~\ref{lem:markov} again yields
	\[
		|P'(0) - f'(0)| = |(P - P_L)'(0)| \leq 64 L^2\eps A_0. \qedhere
	\]
\end{proof}

\subsection{Lower Bound for Subspace Sketches}
By Lemma~\ref{lem:S_size} and a greedy Gilbert-Varshamov packing inside $\cS$, there is a subfamily $\cC\subseteq\cS$ with pairwise Hamming distance at least $\rho rs^2$ and $|\cC|=2^{\Omega(rs^2)}$. 

By Yao's minimax principle, it suffices to consider a deterministic data structure $Q_p$ for the distribution that chooses $S$ uniformly from $\cC$ and gives the rounded hard matrix $\widetilde{A}_S$ as input. Suppose $Q_p$ is correct with probability at least $0.9$ over this choice of $S$, and let $\cC'\subseteq\cC$ be the set of $S$ for which $Q_p$ is correct. Then $|\cC'|\geq 0.9|\cC|$, and $\cC'$ still has pairwise Hamming distance at least $\rho rs^2$. Combining Lemmata~\ref{lem:conditional_recovery} and~\ref{lem:derivative_approximation}, we see that, for each $S\in\cC'$, the deterministic sketch $Q_p$ can recover $S$.
Indeed, the sketch determines functions $\widehat{D}_i$ compatible with $S$, and Lemma~\ref{lem:conditional_recovery}
shows that $S$ is the unique compatible collection in $\cC$. Distinct collections in $\cC'$ therefore require distinct content in the data structure.

By the maximality of $k$ in \eqref{eqn:assumption_k}, we have $k=\Theta_p(\Lambda)$ and $2^k\asymp_p \eps^{-2}\Lambda^{-4}k^{-p-3/2}$. Hence, $r\asymp 2^k/\sqrt{k}\gtrsim_p \eps^{-2}\Lambda^{-p-6}$.
Applying the deterministic counting argument to $\cC'$, the sketch must have at least
\[
	\log|\cC'| = \Omega(rs^2)
	= \Omega_p\mleft(\frac{d^2}{\eps^2\log^{p+6}(2/\eps)}\mright)
\]
bits. Finally, the constructed matrix has $n_0=2^k m_0$ rows. Since $2^k=\widetilde\Theta_p(\eps^{-2})$, $s=\Theta(d)$ and $m_0=O_p((s/\eps^2)^{\alpha_p})$, we have
\[
	n_0=\widetilde O_p(d^{\alpha_p}/\eps^{2+2\alpha_p}).
\]
For any larger number of rows $n$, we pad the hard instance with zero rows. This proves Theorem~\ref{thm:subspace_sketch_lower_bound}.

\section{Dimension Lower Bound}
\begin{proof}[Proof of Theorem~\ref{thm:dimension_lower_bound}]
	We use the standard conversion from sketch lower bounds to dimension lower bounds
	from \cite{LWW21}. 
	
	Consider the entire family $\{\widetilde A_S:S\in\cC\}$ constructed in the proof of
	Theorem~\ref{thm:subspace_sketch_lower_bound} with accuracy parameter $\eps$.
	Let $\gamma=c_p\eps$, where $c_p>0$ is sufficiently small, and let $N$ be the
	smallest integer such that every member of this family admits a linear embedding
	into $\ell_p^N$ satisfying the $(1\pm\gamma)$ norm guarantee.
	Then $N\leq n_0$ by the identity embedding, and $N\leq N_p(d,\gamma)$.
	For each $S\in\cC$, choose such an embedding and write its composition with
	$\widetilde A_S$ as $B_S\in\R^{N\times d}$. By the choice of $c_p$,
	\[
		\norm{B_Sx}_p^p = \left(1\pm\frac{\eps}{4}\right)\norm{\widetilde A_Sx}_p^p,
		\qquad\forall x\in\R^d.
	\]
	By \eqref{eqn:A_lower_bound} and the rounding and embedding guarantees, $\norm{B_Sx}_p\gtrsim_p2^{k/p}\norm{x}_2$. Thus we may round
	$B_S$ to a matrix $\widehat B_S$ on a common dyadic grid of spacing
	$\Theta_p(\eps/(N^{1/p}\sqrt d))$, with the implicit constant sufficiently small, so that
	\[
		\norm{\widehat B_S x}_p^p = \left(1\pm\frac{\eps}{4}\right)\norm{B_Sx}_p^p,
		\qquad\forall x\in\R^d.
	\]
	Thus $Q_p(x)=\norm{\widehat B_Sx}_p^p$ is a $(1\pm\eps)$ sketch for
	$\widetilde A_S$. Every entry of $B_S$ has magnitude at most
	$O_p(n_0^{1/p})$, by applying the embedding guarantee to coordinate vectors
	and using the entry bounds for $\widetilde A_S$.
	Since $N\leq n_0$, storing $\widehat B_S$ therefore uses
	$O_p(Nd\log(n_0d/\eps))$ bits.
	These sketches are valid for every member of the hard family, so the counting
	argument in the proof of Theorem~\ref{thm:subspace_sketch_lower_bound} gives
	\[
		N d \log(n_0 d/\eps)
		\gtrsim_p
		\frac{d^2}{\eps^2\polylog(1/\eps)}.
	\]
	Since $n_0=\widetilde O_p(d^{\alpha_p}/\eps^{2+2\alpha_p})$, it follows that
	\[
		N_p(d,c_p\eps)\geq N
		\gtrsim_p \frac{d}{\eps^2\polylog(d/\eps)}.
	\]
	Rescaling $\eps$ proves the theorem,
	after decreasing $\eps_0(p)$ if necessary.
\end{proof}

\section*{Acknowledgements}
The author is supported in part by the Singapore Ministry of Education AcRF Tier 1 grant RG21/25.

The central technical idea of this paper originated from ChatGPT 6 Pro. The author substantially revised and completed the AI-assisted material, and is responsible for all mathematical claims, proofs, references, and the final manuscript.

\bibliographystyle{plain}
\bibliography{reference}

\end{document}